\documentclass{article}

\usepackage{latexsym}
\usepackage{amssymb}
\usepackage{stmaryrd}
\usepackage{amsmath,amsthm}
\usepackage{bbold}
\usepackage{mathpartir}
\usepackage{proof}
\usepackage{blkarray}
\usepackage{wrapfig}

\theoremstyle{plain}
\newtheorem{theorem}{Theorem}
\newtheorem{lemma}{Lemma}[theorem]

\theoremstyle{definition}
\newtheorem{definition}{Definition}

\theoremstyle{remark}

\newtheorem{example}{Example}[theorem]

\newcommand{\alt}{~\mid~}
\newcommand{\bor}{\alt}
\newcommand{\iso}{\leftrightarrow}

\newcommand{\inl}[1]{\mathtt{inj}_l{\;#1}}
\newcommand{\inr}[1]{\mathtt{inj}_r{\;#1}}

\newcommand{\pv}[2]{\langle #1,#2 \rangle}
\newcommand{\letpv}[3]{{\tt let}\,{\langle #1 \rangle}={#2}~{\tt in}~{#3}}
\newcommand{\letv}[3]{{\tt let}\,{#1}={#2}~{\tt in}~{#3}}
\newcommand{\clause}[2]{\mid ~ #1 \iso #2}
\newcommand{\clauses}[1]{\left\{ ~#1~ \right\}}
\newcommand{\boolt}{\mathbb{B}}
\newcommand{\tc}{\mathtt{t}\!\mathtt{t}}
\newcommand{\fc}{\mathtt{f}\!\mathtt{f}}
\newcommand{\isoterm}{\omega}
\newcommand{\OD}[2]{{\rm OD}_{#1}{#2}}
\newcommand{\ODe}[2]{{\rm OD}^{\it ext}_{#1}{#2}}
\newcommand{\BV}{{\rm Val}}
\newcommand{\FV}{{\rm FV}}

\newcommand{\entailval}{\vdash_v}
\newcommand{\entailiso}{\vdash_{\isoterm}}

\newcommand{\scalprod}[2]{\langle#1|#2\rangle}
\newcommand{\tensor}{\otimes}
\newcommand{\Cx}{\mathbb{C}}
\newcommand{\denot}[1]{{\llbracket{#1}\rrbracket}}
\newcommand{\base}[1]{{\mathcal{B}_{#1}}}

\allowdisplaybreaks

\title{From Symmetric Pattern-Matching to \\ Quantum Control\\ (Extended
Version)}

\author{Amr Sabry
\and Beno\^{i}t Valiron${}^{*}$
\and Juliana Kaizer Vizzotto\thanks{Partially funded by FoQCoss STIC AmSud project - STIC-AmSUD/Capes -
Foundations of Quantum Computation: Syntax and Semantics.
}}

\begin{document}
\maketitle 

\begin{abstract}
  One perspective on quantum algorithms is that they are classical
  algorithms having access to a special kind of memory with exotic
  properties. This perspective suggests that, even in the case of
  quantum algorithms, the control flow notions of sequencing,
  conditionals, loops, and recursion are entirely classical. There is
  however, another notion of control flow, that is itself quantum. The
  notion of quantum conditional expression is reasonably
  well-understood: the execution of the two expressions becomes itself
  a superposition of executions. The quantum counterpart of loops and
  recursion is however not believed to be meaningful in its most
  general form.

  In this paper, we argue that, under the right circumstances, a
  reasonable notion of quantum loops and recursion is possible. To
  this aim, we first propose a classical, typed, reversible language
  with lists and fixpoints. We then extend this language to the
  \emph{closed} quantum domain (without measurements) by allowing
  linear combinations of terms and restricting fixpoints to
  structurally recursive fixpoints whose termination proofs match the
  proofs of convergence of sequences in infinite-dimensional Hilbert
  spaces. We additionally give an operational semantics for the
  quantum language in the spirit of algebraic lambda-calculi and
  illustrate its expressiveness by modeling several common unitary
  operations.
\end{abstract}

\section{Introduction}

The control flow of a program describes how its elementary operations
are organized along the execution. Usual primitive control mechanisms
are sequences, tests, iteration and recursion. Elementary operations
placed in sequence are executed in order. Tests allow conditionally
executing a group of operations and changing the course of the
execution of the program. Finally, iteration gives the possibility to
iterate a process an arbitrary number of times and recursion
generalizes iteration to automatically manage the history of the
operations performed during iteration. The structure of control flow
for conventional (classical) computation is well-understood. In the
case of \emph{quantum} computation, control flow is still subject to
debate. This paper proposes a working notion of quantum control in
closed quantum systems, shedding new light on the problem, and
clarifying several of the previous concerns.

\paragraph{Quantum computation.} 
A good starting point for understanding quantum computation is to
consider classical circuits over \emph{bits} but replacing the bits
with \emph{qubits}, which are intuitively superpositions of bits
weighed by complex number amplitudes. Computationally, a qubit is an
abstract data type governed by the laws of quantum physics, whose
values are normalized vectors of complex numbers in the Hilbert space
$\Cx^2$ (modulo a global phase). By choosing an orthonormal basis, say
the classical bits $\tc$ and $\fc$, a qubit can be regarded as a
complex linear combination, $\alpha~\tc + \beta~\fc$, where $\alpha$
and $\beta$ are complex numbers such that
$|\alpha|^2+|\beta|^2=1$. This generalizes naturally to multiple
qubits: the state of a system of $n$ qubits is a vector in the Hilbert
space $(\Cx^2)^{\otimes{}n}$.

The operations one can perform on a quantum memory are of two kinds:
quantum gates and measurements. Quantum gates are unitary operations
that are ``purely quantum'' in the sense that they modify the quantum
memory without giving any feedback to the outside world: the quantum
memory is viewed as a {\em closed system}. A customary graphical
representation for these operations is the {\em quantum circuit}, akin
to conventional boolean circuits: wires represent qubits while boxes
represents operations to perform on them. One of the peculiar aspects
of quantum computation is that the state of a qubit is
non-duplicable~\cite{wootters82single}, a result known as the {\em
  no-cloning theorem}. A corollary is that a quantum circuit is a very
simple kind of circuit: wires neither split nor merge.

Measurement is a fundamentally different kind of operation: it queries
the state of the quantum memory and returns a classical
result. Measuring the state of a quantum bit is a probabilistic and
destructive operation: it produces a classical answer with a
probability that depends on the amplitudes $\alpha, \beta$ in the
state of the qubit while projecting this state onto $\tc$ or $\fc$,
based on the result.

For a more detailed introduction to quantum computation, we refer the
reader to recent textbooks (e.g., ~\cite{NielsenChuang}).

\paragraph{Control flow in quantum computation.}
In the context of quantum programming languages, there is a
well-understood notion of control flow: the so-called {\em classical
  control flow}. A quantum program can be seen as the construction,
manipulation and evaluation of quantum circuits~\cite{quipper,qwire}.
In this setting, circuits are simply considered as special kinds of
data without much computational content, and programs are ruled by
regular classical control.

One can however consider the circuit being manipulated as a program in
its own right: a particular sequence of execution on the quantum
memory is then seen as a closed system. One can then try to derive a
notion of {\em quantum control}~\cite{qml}, with ``quantum tests'' and
``quantum loops''. Quantum tests are a bit tricky to
perform~\cite{qml,qalternation} but they essentially correspond to
well-understood controlled operations. The situation with quantum
loops is more subtle~\cite{qalternation,yingbook}. First, a
hypothetical quantum loop {\em must} terminate. Indeed, a
non-terminating quantum loop would entail an infinite quantum circuit,
and this concept has so far no meaning. Second, the interaction of
quantum loops with measurement is problematic: it is known that the
canonical model of \emph{open} quantum computation based on
superoperators~\cite{selinger04quantum,qarrow} is incompatible with
such quantum control~\cite{qalternation}. Finally, the mathematical
operator corresponding to a quantum loop would need to act on an
infinite-dimensional Hilbert space and the question of mixing
programming languages with infinitary Hilbert spaces is still an
unresolved issue.

\paragraph{Our contribution.}
In this paper, we offer a novel solution to the question of quantum
control: we define a purely quantum language, inspired by
Theseus~\cite{theseus}, featuring tests and fixpoints together with
lists. More precisely, we propose (1) a typed, reversible language,
extensible to linear combinations of terms, with a reduction strategy
akin to algebraic
lambda-calculi~\cite{linvec,lineal,Vaux09};
(2) a model for the language based on unitary operators over
infinite-dimensional Hilbert spaces, simplifying the Fock space
model of Ying~\cite{yingbook}. This model captures lists, tests,
and structurally recursive fixpoints. We therefore settle two
longstanding issues.  (1) We offer a solution to the problem of
quantum loops, with the use of {\em terminating}, {\em structurally
  recursive}, {\em purely quantum} fixpoints. We dodge previously
noted concerns (e.g.,~\cite{qalternation}) by staying in the closed
quantum setting and answer the problem of the external system of quantum
``coins''~\cite{yingbook} with the use of lists.
(2) By using a linear language based on patterns and
clauses, we give an extensible framework for reconciling algebraic
calculi with quantum computation~\cite{tonder04lambda,lineal,linvec}.

In the remainder of the paper, we first introduce the key idea
underlying our classical reversible language in a simple first-order
setting. We then generalize the setting to allow second-order
functions, recursive types (e.g., lists), and fixpoints. After
illustrating the expressiveness of this classical language, we adapt
it to the quantum domain and give a semantics to the resulting quantum
language in infinite-dimensional Hilbert spaces.

This technical report is an extended version of a paper accepted for
publication in the proceedings of FoSSaCS'18~\cite{shortversion}.

\section{Pattern-Matching Isomorphisms}
\label{sec:intro-iso}
\label{sec:iso-1st-order}
 
The most elementary control structure in a programming language is the
ability to conditionally execute one of several possible code
fragments. Expressing such an abstraction using predicates and nested
\textbf{if}-expressions makes it difficult for both humans and
compilers to reason about the control flow structure. Instead, in
modern functional languages, this control flow paradigm is elegantly
expressed using \emph{pattern-matching}. This approach yields code
that is not only more concise and readable but also enables the
compiler to easily verify two crucial properties: (i) non-overlapping
patterns and (ii) exhaustive coverage of a datatype using a collection
of patterns. Indeed most compilers for functional languages perform
these checks, warning the user when they are violated. At a more
fundamental level, e.g., in type theories and proof assistants, these
properties are actually necessary for correct reasoning about
programs. Our first insight, explained in this section, is that these
properties, perhaps surprisingly, are sufficient to produce a simple
and intuitive first-order reversible programming language.

\begin{figure}[t]
  \centering
  \begin{minipage}{0.4\linewidth}
    \begin{verbatim}
f :: Either Int Int -> a
f (Left 0)     = undefined
f (Left (n+1)) = undefined
f (Right n)    = undefined
\end{verbatim}\vspace{-4ex}
    \caption{A skeleton}\label{fig:intro-ex-1}
  \end{minipage}
  \hfill
  \begin{minipage}{0.4\linewidth}
\begin{verbatim}
g :: (Bool,Int) -> a
g (False,n)  = undefined
g (True,0)   = undefined
g (True,n+1) = undefined
\end{verbatim}\vspace{-4ex}
    \caption{Another skeleton}\label{fig:intro-ex-2}
  \end{minipage}
  \\[3ex]
  \begin{minipage}{0.6\linewidth}
\begin{verbatim}
h :: Either Int Int <-> (Bool,Int)
h (Left 0)     = (True,0)
h (Left (n+1)) = (False,n)
h (Right n)    = (True,n+1)
\end{verbatim}\vspace{-4ex}
    \caption{An isomorphism}\label{fig:intro-ex-3}
  \end{minipage}
\end{figure}

\subsection{An Example}

We start with a small illustrative example, written in a Haskell-like
syntax.  Fig.~\ref{fig:intro-ex-1} gives the skeleton of a function
\verb|f| that accepts a value of type \verb|Either Int| \verb|Int|; the
patterns on the left-hand side exhaustively cover every possible
incoming value and are non-overlapping. Similarly,
Fig.~\ref{fig:intro-ex-2} gives the skeleton for a function~\verb|g|
that accepts a value of type \verb|(Bool,Int)|; again the patterns on
the left-hand side exhaustively cover every possible incoming value
and are non-overlapping. Now we claim that since the types
\verb|Either Int Int| and \verb|(Bool,Int)| are isomorphic, we can
combine the patterns of \verb|f| and \verb|g| into \emph{symmetric
  pattern-matching clauses} to produce a reversible function between
the types \verb|Either Int Int| and
\verb|(Bool,Int)|. Fig.~\ref{fig:intro-ex-3} gives one such function;
there, we suggestively use \verb|<->| to indicate that the function
can be executed in either direction. This reversible function is
obtained by simply combining the non-overlapping exhaustive patterns
on the two sides of a clause. In order to be well-formed in either
direction, these clauses are subject to the constraint that each
variable occurring on one side must occur exactly once on the other
side (and with the same type). Thus it is acceptable to swap the
second and third right-hand sides of \verb|h| but not the first and
second ones.

\subsection{Terms and Types}
\label{sec:1st-order-finite}

We present a formalization of the ideas presented above using a simple
typed first-order reversible language. The language is two-layered.
The first layer contains values, which also play the role of
patterns. These values are constructed from variables ranged over 
$x$ and the introduction forms for the finite types $a,b$ constructed
from the unit type and sums and products of types. The second layer
contains collections of pattern-matching clauses denoting
isomorphisms of type $a \iso b$. Computations are chained applications
of isomorphisms to values:

\begin{alignat*}{10}
&\text{(Value types)}\quad & a, b &~~&&::= ~&&   \mathbb{1} \alt a \oplus b
                                    \alt a \otimes b\\
&\text{(Iso types)} & T &&&::=&& a \iso b \\[1.5ex]
&\text{(Values)} & v &&&::=&& () \alt x \alt \inl{v} \alt \inr{v} \alt
                            \pv{v_1}{v_2}\\
&\text{(Isos)} & \isoterm &&&::=&& 
           \clauses{\clause{v_1}{v'_1}\clause{v_2}{v'_2}~\ldots}
                              \\
&\text{(Terms)} & t &&& ::= && v \alt \isoterm\,t
\end{alignat*}

The typing rules are defined using two judgments:
$\Delta \entailval v : a$ for typing values (or {\em patterns}) and
terms; and $\entailiso \isoterm : a \iso b$ for typing collections of
pattern-matching clauses denoting an isomorphism. As it is customary,
we write $a_1\tensor a_2\tensor\cdots\tensor a_n$ for
$((a_1\tensor a_2)\tensor\cdots\tensor a_n)$, and similarly
$\pv{x_1}{x_2,\ldots,x_n}$ for $\pv{\pv{x_1}{x_2}}{\ldots,x_n}$.

The typing rules for values are the expected ones. The only subtlety
is the fact that they are linear: because values act as patterns, we
forbid the repetition of variables. A typing context $\Delta$ is a set
of typed variables $x_1:a_1,\ldots, x_n:a_n$. A value typing judgment
is valid if it can be derived from the following rules:

\[\begin{array}{c}
\infer{
  \entailval() : \mathbb{1},
}{}
\qquad
\infer{
  x:a\entailval x:a,
}{}
\qquad
\infer{
  \Delta_1,\Delta_2\entailval\pv{v_1}{v_2} : a\otimes b.
}{
  \Delta_1\entailval v_1 : a
  &
  \Delta_2\entailval v_2 : b
}
\\ \\
\infer{
  \Delta\entailval\inl{v} : a\oplus b,
}{
  \Delta\entailval v : a
}
\qquad
\infer{
  \Delta\entailval\inr{v} : a\oplus b,
}{
  \Delta\entailval v : b
}
\end{array}\]

\noindent The typing rule for term construction is simple and forces
the term to be closed:
\[
\infer{
  \entailval \isoterm~t : b
}{
  \entailval t : a & \entailiso \isoterm : a \iso b
}
\]

\noindent The most interesting type rule is the one for
isomorphisms. We present the rule and then explain it in detail:

\begin{equation}\label{eq:typ-iso-specialized}
\infer{ 
  \entailiso 
  \clauses{\clause{v_1}{v'_1}\clause{v_2}{v'_2}~\ldots} : a \iso b,
}{
  \begin{array}{@{}l@{}}
    \Delta_1\entailval v_1 : a 
    \\
    \Delta_1\entailval v'_1 : b
  \end{array}
  &
  \ldots 
  &
  \begin{array}{@{}l@{}}
    \Delta_n\entailval v_n : a 
    \\
    \Delta_n\entailval v'_n : b
  \end{array}
  &
  \begin{array}{@{}l@{}}
    \forall i\neq j, v_i\bot v_j
    \\
    \forall i\neq j, v'_i\bot v'_j
  \end{array}
  &
  \begin{array}{@{}l@{}}
    \dim(a) = n
    \\
    \dim(b) = n
  \end{array}
}
\end{equation}

\noindent The rule relies on two auxiliary conditions as motivated in
the beginning of the section. These conditions are (i) the
orthogonality judgment $v \bot v'$ that formalizes that patterns must
be \emph{non-overlapping} and (ii) the condition $\dim(a)=n$ which
formalizes that patterns are \emph{exhaustive}. The rules for deriving
orthogonality of values or patterns are:
\[
  \begin{array}{c}
    \infer{\inl{v_1}~\bot~\inr{v_2}}{}
    \qquad
    \infer{\inr{v_1}~\bot~\inl{v_2}}{}
    \\[2ex]
    \infer{\inl{v_1}~\bot~\inl{v_2}}{v_1~\bot~v_2}
    \quad
    \infer{\inr{v_1}~\bot~\inr{v_2}}{v_1~\bot~v_2}
    \quad
    \infer{\pv{v}{v_1}~\bot~\pv{v'}{v_2}}{v_1~\bot~v_2}
    \quad
    \infer{\pv{v_1}{v}~\bot~\pv{v_2}{v'}}{v_1~\bot~v_2}
  \end{array}
\]
\noindent The idea is simply that the left and right injections are
disjoint subspaces of values. To characterize that a set of patterns
is exhaustive, we associate a \emph{dimension} with each type. For
finite types, this is just the number of elements in the type and is
inductively defined as follows: $\dim(\mathbb{1})=1$;
$\dim(a\oplus b) = \dim(a)+\dim(b)$; and
$\dim(a\otimes b) = \dim(a)\cdot\dim(b)$. For a given type $a$, if a
set of non-overlapping clauses has cardinality $\dim(a)$, it is
exhaustive. Conversely, any set of exhaustive clauses for a type $a$
either has cardinality $\dim(a)$ or can be extended to an equivalent
exhaustive set of clauses of cardinality $\dim(a)$.

\subsection{Semantics} 
\label{sec:iso-iso}

We equip our language with a simple operational semantics on terms,
using the natural notion of matching.  To formally define it, we first
introduce the notion of variable assignation, or valuation, which is a
partial map from a finite set of variables (the support) to a set of
values.  We denote the matching of a value $w$ against a pattern $v$
and its associated valuation $\sigma$ as $\sigma[v] = w$ and define it
as follows:
\[
\infer{\sigma[()] = ()}{}
\quad
\infer{\sigma[x] = v}{\sigma = \{ x \mapsto v\}}
\quad
\infer{\sigma[\inl{v}] = \inl{w}}{\sigma[v] = w}
\quad
\infer{\sigma[\inr{v}] = \inr{w}}{\sigma[v] = w}
\]
\[
\infer{
  \sigma[\pv{v_1}{v_2}] = \pv{w_1}{w_2}
}{
  \sigma_2[v_1] = w_1
  &
  \sigma_1[v_2] = w_2
  &
  \text{supp}(\sigma_1) \cap \text{supp}(\sigma_2) = \emptyset
  &
  \sigma = \sigma_1\cup\sigma_2
}
\]
If $\sigma$ is a valuation whose support contains the variables of
$v$, we write $\sigma(v)$ for the value where the variables of $v$
have been replaced with the corresponding values in $\sigma$:
\begin{itemize}
\item $\sigma(()) = ()$
\item $\sigma(x) = v$ is $\{x\mapsto v\}\subseteq \sigma$
\item $\sigma(\inl{v}) = \inl{\sigma(v)}$
\item $\sigma(\inr{v}) = \inr{\sigma(v)}$
\item $\sigma(\pv{v_1}{v_2}) = \pv{\sigma(v_1)}{\sigma(v_2)}$
\end{itemize}

\noindent Given these definitions, we can define the reduction
relation on terms. The redex
$ \clauses{\clause{v_1}{v'_1}\clause{v_2}{v'_2}~\ldots} v $ reduces to
$\sigma(v'_i)$ whenever $\sigma[v_i] = v'_i$. Because of the
conditions on patterns, a matching pattern exists by exhaustivity of
coverage, and this pattern is unique by the non-overlapping
condition. Congruence holds: $\isoterm\,t\to\isoterm\,t'$ whenever
$t\to t'$.  As usual, we write $s\to t$ to say that~$s$ rewrites in
one step to $t$ and $s\to^*t$ to say that~$s$ rewrites to $t$ in~0 or
more steps.

Because of the conditions set on patterns, the rewrite system is
deterministic. More interestingly, we can swap the two sides of all
pattern-matching clauses in an isomorphism $\isoterm$ to get
$\isoterm^{-1}$. The execution of $\isoterm^{-1}$ is the reverse
execution of $\isoterm$ in the sense that
$\isoterm^{-1}(\isoterm~t) \to^* t $ and
$\isoterm(\isoterm^{-1}~t') \to^* t'$.

\section{Second-Order Functions, Lists, and Recursion}

The first-order reversible language from the previous section embodies
symmet\-ric-pattern matching clauses as its core notion of
control. Its expressiveness is limited, however. We now show that it
is possible to extend it to have more in common with a conventional
functional language. To that end, we extend the language with the
ability to parametrically manipulate isomorphisms, with a recursive
type (lists), and with recursion.

\subsection{Terms and Types}
\label{sec:2st-order}

Formally, the language is now defined as follows.
\begin{alignat*}{100}
&\text{(Val }\& ~\text{term types)} \quad& a, b &&&::=~ &&  
                        \mathbb{1} \alt a \oplus b \alt a \otimes b \alt [a]
                          \\
&\text{(Iso types)} & T &&&::=&&  a \iso b \alt (a \iso b) \to T \\[1.5ex]
&\text{(Values)} & v &&&::=&&
                       () \alt x \alt \inl{v} \alt \inr{v} \alt
                        \pv{v_1}{v_2}\\
&\text{(Products)} & p &&&::=&&
                       () \alt x \alt \pv{p_1}{p_2}\\
&\text{(Extended Values)} & e &&&::=&& v \alt
                                 \letv{p_1}{\isoterm~p_2}{e}\\
&\text{(Isos)} & \isoterm &&&::=&& 
                          \clauses{\clause{v_1}{e_1}\clause{v_2}{e_2}~\ldots}
                          \alt \lambda f.\isoterm \alt \\
&            &          &&&   && \mu f.\isoterm \alt f \alt 
                                          \isoterm_1\,\isoterm_2 \\
&\text{(Terms)}  & t &&&::=&& () \alt x \alt \inl{t} \alt \inr{t} \alt
                        \pv{t_1}{t_2} \alt \\ 
&              &   &&&   && \isoterm~t \alt \letv{p}{t_1}{t_2}
\end{alignat*}

\noindent We use variables $f$ to span a set of iso-variables and
variables $x$ to span a set of term-variables.  We extend the layer of
isos so that it can be parameterized by a fixed number of other isos,
i.e., we now allow higher-order manipulation of isos using
$\lambda f.\isoterm$, iso-variables, and applications.  Isos can now
be used inside the definition of other isos with a let-notation. These
let-constructs are however restricted to products of term-variables:
they essentially serve as syntactic sugar for composition of isos. An
extended value is then a value where some of its free variables are
substituted with the result of the application of one or several isos.
Given an extended value $e$, we define its {\em bottom value}, denoted
with $\BV(e)$ as the value ``at the end'' of the let-chain:
$\BV(v) = v$, and $\BV(\letv{p}{\isoterm p}{e})=\BV(e)$. The
orthogonality of extended values is simply the orthogonality of their
bottom value.

\begin{table}[tb]
\[
\begin{array}{c}
\infer{\emptyset;\Psi\entailval () : \mathbb{1}}{}
\qquad
\infer{x:a;\Psi \entailval x :a}{}
\\ \\
\infer{\Delta;\Psi\entailval \inl{t}:a \oplus b}{\Delta;\Psi\entailval
  t:a}
\qquad
\infer{\Delta;\Psi\entailval\inr{t}:a \oplus b}{\Delta;\Psi\entailval t:b}
\\ \\
\infer{
  \Delta_1,\Delta_2;\Psi\entailval \pv{t_1}{t_2} : a \otimes b
}{
  \Delta_1;\Psi\entailval t_1 : a 
  & 
  \Delta_2;\Psi\entailval t_2 : b
}
\\ \\
\infer{
  \Delta;\Psi\entailval \isoterm~t : b
}{
  \Psi\entailiso \isoterm : a \iso b
  & 
  \Delta;\Psi\entailval t : a
}
\quad
\infer{
  \Delta;\Psi\entailval \letpv{x,y}{t_1}{t_2} : c
}{
  \Delta;\Psi\entailval t_1 : a \otimes b
  & 
  \Delta,x:a,y:b;\Psi\entailval t_2 : c
}
\end{array}
\]
\caption{Typing rules for terms and values}
\label{tab:termtyp}
\end{table}

\begin{table}[t]
\[
\begin{array}{c}
\infer{
  \Psi\entailiso 
  \clauses{\clause{v_1}{e_1}\clause{v_2}{e_2}~\ldots} : a \iso b.
}{
  \begin{array}{l@{\quad}l@{\quad}l@{\qquad}l}
    \Delta_1;\Psi \entailval v_1 : a 
    &
      \ldots
    &
      \Delta_n;\Psi\entailval v_n : a
    & 
      \OD{a}\{v_1,\ldots,v_n\}
    \\
    \Delta_1;\Psi \entailval e_1 : b 
    &
      \ldots
    &
      \Delta_n;\Psi\entailval e_n : b
    & \ODe{b}\{e_1,\ldots,e_n\}
  \end{array}
}
\\ \\
\infer{
  \Psi \entailiso \lambda f.\isoterm :(a\iso b) \to T
}{
  \Psi,f:a\iso b \entailiso \isoterm : T
}
\qquad
\infer{ \Psi,f : T \entailiso f : T}{}
\quad \\ \\
\infer{\Psi \entailiso \isoterm_1 \isoterm_2 : T}
{\Psi \entailiso \isoterm_1 : (a\iso b) \to T \;\; 
\Psi \entailiso \isoterm_2 : a\iso b}
\\ \\
\inferrule{
  \Psi, f:a\iso b\entailiso \isoterm : (a_1\iso b_1)\to\cdots \to(a_n\iso
  b_n)\to(a\iso b)
  \\
  \text{$\mu f.\isoterm$ terminates in any finite context}
}
{
  \Psi\entailiso \mu f.\isoterm : (a_1\iso b_1)\to\cdots \to(a_n\iso b_n)\to(a\iso b)
}
\\
\end{array}
\]
\caption{Typing rules for isos}
\label{tab:isotyp}
\end{table}

As usual, the type of lists $[a]$ of elements of type $a$ is a
recursive type and is equivalent to $\mathbb{1}\oplus(a\times[a])$.
We build the value $[]$ (empty list) as $\inl{()}$ and the term
$t_1:t_2$ (cons of $t_1$ and $t_2$) as $\inr{\pv{t_1}{t_2}}$.  In
addition, to take full advantage of recursive datatypes, it is natural
to consider recursion. Modulo a termination guarantee it is possible
to add a fixpoint to the language: we extend isos with the fixpoint
constructor $\mu f.\isoterm$. Some reversible languages allow infinite
loops and must work with partial isomorphisms instead. Since we plan
on using our language as a foundation for a quantum language we insist
of termination.

Since the language features two kinds of variables, there are typing
contexts (written $\Delta$) consisting of base-level typed variables
of the form $x:a$, and typing context (written $\Psi$) consisting of
typed iso-variables of the form $f:T$. As terms and values contain
both base-level and iso-variables, one needs two typing
contexts. Typing judgments are therefore written respectively as
$\Delta;\Psi\entailval t:a$.  The updated rules for $(\entailval)$ are
found in Tab.~\ref{tab:termtyp}. As the only possible free variables
in isos are iso-variables, their typing judgments only need one
context and are written as $\Psi\entailiso \isoterm:T$.

The rules for typing derivations of isos are in
Tab.~\ref{tab:isotyp}. It is worthwhile mentioning that isos are
treated in a usual, non-linear way: this is the purpose of the typing
context separation. The intuition is that an iso is the description of
a closed computation with respect to inputs: remark that isos cannot
accept value-types. As computations, they can be erased or duplicated
without issues. On the other hand, value-types still need to be
treated linearly.

In the typing rule for recursion, the condition ``$\mu f.\isoterm$
terminates in any finite context'' formally refers to the following
requirement. A well-typed fixpoint $\mu f.\isoterm$ of type
$\Psi\entailiso \mu f.\isoterm : (a_1\iso b_1)\to\cdots \to(a_n\iso
b_n)\to(a\iso b)$ is {\em terminating in a $0$-context} if for all
closed isos $\isoterm_i:a_i\iso b_i$ not using fixpoints and for every
closed value $v$ of type $a$, the term
$((\mu f.\isoterm)\isoterm_1\ldots\isoterm_n)v$ terminates.  We say
that the fixpoint is {\em terminating in an $(n+1)$-context} if for
all closed isos $\isoterm_i:a_i\iso b_i$ terminating in $n$-contexts,
and for every closed value $v$ of type $a$, the term
$((\mu f.\isoterm)\isoterm_1\ldots\isoterm_n)v $ terminates.  Finally,
we say that the fixpoint is {\em terminating in any finitary context}
if for all $n$ it is terminating in any $n$-context.

With the addition of lists, the non-overlapping and exhaustivity
conditions need to be modified. The main problem is that we can no
longer define the dimension of types using natural numbers: $[a]$ is
in essence an infinite sum, and would have an ``infinite''
dimension. Instead, we combine the two conditions into the concept of
\emph{orthogonal decomposition}.  Formally, given a type $a$, we say
that a set $S$ of patterns is an {\em orthogonal decomposition},
written $\OD{a}{(S)}$, when these patterns are pairwise orthogonal and
when they cover the whole type.  We formally define $\OD{a}{(S)}$ as
follows.  For all types $a$, $\OD{a}\{x\}$ is valid.  For the unit
type, $\OD{\mathbb{1}}\{()\}$ is valid.  If $\OD{a}(S)$ and
$\OD{b}(T)$, then
\begin{align*}
  &\OD{a\oplus b}(\{\inl{v}~|~ v\in S\}\cup\{\inr{v}~|~ v\in T\})
  \\
  \text{and}\quad
  &\OD{a\otimes b}\{\pv{v_1}{v_2} ~|~ v_1\in S,~ v_2\in T,~
  \FV(v_1)\cap\FV(v_2)=\emptyset\}, 
\end{align*}
where $\FV(t)$ stands for the set of free value-variables in $t$. We
then extend the notion of orthogonal decomposition to extended values
as follows. If $S$ is a set of extended values, $\ODe{a}(S)$ is true
whenever $ \OD{a}\{\BV(e) ~|~ e \in S\} $. With this new
characterization, the typing rule of iso in
Eq.~\ref{eq:typ-iso-specialized} still holds, and then can be
re-written using this notion of orthogonal decomposition as shown in
Tab.~\ref{tab:isotyp}.

\subsection{Semantics} 
\label{subsec:2st-semantics}

In Tab.~\ref{tab:reduction} we present the reduction rules for the
reversible language. We assume that the reduction relation applies to
well-typed terms. In the rules, the notation $C[-]$ stands for an {\em
  applicative context}, and is defined as follows.
\[
  \begin{array}{lll}
    C[-]&::=
    & [-] \bor
      \inl{C[-]} \bor
      \inr{C[-]} \bor
      (C[-])\isoterm \bor
    \\
        &&\{\cdots\}~C[-] \bor
           \letv{p}{C[-]}{t_2}\bor
           \pv{C[-]}{v} \bor
           \pv{v}{C[-]}.
  \end{array}
\]

\begin{table}[t]
\[
\begin{array}{c}
\infer[\mathrm{Cong}]{C[t_1] \to C[t_2]}{t_1 \to t_2}
\qquad
\infer[\mathrm{LetE}]{\letv{p}{v_1}{t_2} \to \sigma(t_2)}{\sigma[p] = v_1}
\\[1.5ex]
\infer[\mathrm{IsoApp}]{ \clauses{\clause{v_1}{t_1}\;|~\ldots\;
  \clause{v_n}{t_n} } \; v \to \sigma(t_i)}{
  \sigma[v_i] = v}
\\[1.5ex]
\infer[\mathrm{HIsoApp}]{(\lambda f.\isoterm)\;\isoterm_2 \to \isoterm[\isoterm_2/f]}{}
\\[1.5ex]
\infer[\mathrm{IsoRec}]{
  \mu f.\isoterm \to \lambda f_1\ldots
  f_n.(\isoterm[((\mu f.\isoterm)f_1\ldots f_n)/f])f_1\ldots f_n
}{
  \Psi, f:a\iso b\entailiso\isoterm : (a_1\iso b_1)\to\cdots \to(a_n\iso
  b_n)\to(a\iso b)
}
\end{array}
\]
\caption{Reduction rules}
\label{tab:reduction}
\end{table}

The inversion of isos is still possible but more subtle than in the
first-order case. We define an inversion operation $(-)^{-1}$ on iso
types with, $(a\iso b)^{-1} := (b\iso a)$,
$((a\iso b)\to T)^{-1} := ((b\iso a)\to (T^{-1}))$.  Inversion of isos
is defined as follows. For fixpoints,
$(\mu f.\isoterm)^{-1} = \mu f.(\isoterm^-1)$.  For variables,
$(f)^{-1} := f$.  For applications,
$(\isoterm_1~\isoterm_2)^{-1} := (\isoterm_1)^{-1}~(\isoterm_2)^{-1}$.
For abstraction,
$(\lambda f.\isoterm)^{-1} := \lambda f.(\isoterm^{-1})$. Finally,
clauses are inverted as follows:
\[\begin{array}{c}
      \left(
      \begin{array}{l@{~}c@{~}l}
        v_1&{\iso}&{\tt let}\,p_1=\isoterm_1\,p'_1\,{\tt in}
        \\
           && \cdots
        \\
           && {\tt let}\,p_n=\isoterm_n\,p'_n\,{\tt in}~v'_1
      \end{array}
      \right)^{-1}
      :=
      \left(
      \begin{array}{lcl@{}l@{}l}
        v'_1&{\iso}&{\tt let}\,p'_n&=\isoterm_n^{-1}&\,p_n\,{\tt in}
        \\
           && \cdots
        \\
           && {\tt let}\,p'_1&=\isoterm_1^{-1}&\,p_1\,{\tt in}~v_1
      \end{array}
      \right).
    \end{array}   
  \]
Note that $(-)^{-1}$ only inverts first-order arrows
$(\iso)$, not second-order arrows $(\to)$. This is reflected by the
fact that iso-variable are non-linear while value-variables are. This
is due to the clear separation of the two layers of the language.

The rewriting system satisfies the usual properties for well-typed
terms: it is terminating, well-typed closed terms have a unique normal
value-form, and it preserves typing.

\begin{theorem}\label{th:para-iso-inv-type}
The inversion operation is well-typed, in the sense that if
$
f_1:a_1\iso b_1,\ldots,f_n:a_n\iso b_n
\entailiso
\isoterm : T
$
then we also have
$
f_1:b_1\iso a_1,\ldots,f_n:b_n\iso a_n
\entailiso
\isoterm^{-1} : T^{-1}
$.
\qed
\end{theorem}

Thanks to the fact that the language is terminating, we also recover
the operational result of Sec.~\ref{sec:iso-iso}.

\begin{theorem}\label{th:para-iso-iso}
  Consider a well-typed, closed iso $\entailiso \isoterm:a\iso b$,
  and suppose that $\entailval v:a$ and that $\entailval w:b$, then
  $\isoterm^{-1}(\isoterm~v) \to^* v$ and $\isoterm(\isoterm^{-1}~w) \to^* w$.  \qed
\end{theorem}

\section{Examples}
\label{subsec:2st-examples}

In the previous sections, we developed a novel classical reversible
language with a familiar syntax based on pattern-matching. The
language includes a limited notion of higher-order functions and
(terminating) recursive functions. We illustrate the expressiveness of
the language with a few examples and motivate the changes and
extensions needed to adapt the language to the quantum domain.

We encode booleans as follows:
$\boolt = \mathbb{1} \oplus \mathbb{1}$, $\tc = \inl{()}$, and
$\fc = \inr{()}$. One of the easiest function to define is
${\tt not}:\boolt\iso\boolt$ which flips a boolean. The controlled-not
gate which flips the second bit when the first is true can also be
expressed:
\[
  \begin{array}{l} 
  \mathtt{not} : \boolt \iso \boolt =
        \left(\begin{array}{r@{~\iso~}l}
                {\fc} & {\tc} \\
                {\tc} & {\fc}
         \end{array}\right),
   \end{array}
 \]
 \[
   \begin{array}{l}
   \mathtt{cnot} : \boolt \otimes \boolt \iso \boolt \otimes \boolt = 
     \left(\begin{array}{r@{~\iso~}l}
     \pv{\fc}{x} & \pv{\fc}{x} \\
     \pv{\tc}{\fc} & \pv{\tc}{\tc} \\
     \pv{\tc}{\tc} & \pv{\tc}{\fc}
     \end{array}\right).
   \end{array}
 \]
All the patterns in the previous two functions are orthogonal decompositions which
guarantee reversibility as desired. 

By using the abstraction facilities in the language, we can define
higher-order operations that build complex reversible functions from
simpler ones. For example, we can define a conditional expression
parameterized by the functions used in the two branches:

\[\begin{array}{l}
\mathtt{if} : (a \iso b) \to (a \iso b) \to (\mathbb{B}\otimes a \iso \mathbb{B} \otimes b) \\
\mathtt{if} = 
\lambda g. \lambda h.\left(
  \begin{array}{r@{~\iso~}l}
  \pv{\tc}{x}  & \letv{y}{g~x}{\pv{\tc}{y}} \\
    \pv{\fc}{x}  & \letv{y}{h~x}{\pv{\fc}{y}}
  \end{array}\right)
\end{array}\]

\noindent Using $\mathtt{if}$ and the obvious definition for the
identity function $\mathtt{id}$, we can define
${\tt ctrl} :: (a\iso a) \to (\boolt \otimes a \iso \boolt \otimes a)$
as ${\tt ctrl}~f = \mathtt{if}~f~{\tt id}$ and recover an alternative
definition of {\tt cnot} as ${\tt ctrl}~{\tt not}$. We can then define
the controlled-controlled-not gate (aka the Toffoli gate) by writing
${\tt ctrl}~{\tt cnot}$. We can even iterate this construction using
fixpoints to produce an $n$-controlled-not function that takes a list
of $n$ control bits and a target bit and flips the target bit iff all
the control bits are $\tc$:
\[
\begin{array}{l}
\mathtt{cnot*} : ([\mathbb{B}]\otimes\mathbb{B}) \iso ([\mathbb{B}]\otimes\mathbb{B}) \\
\mathtt{cnot*} = \mu f.\left(
  \begin{array}{r@{~}c@{~}l}
  \pv{[]}{tb}  & {}{\iso}{} & \letv{tb'}{\mathtt{not}~tb}{\pv{[]}{tb'}} \\
  \pv{\fc:cbs}{tb}  & {}{\iso}{} & \pv{\fc:cbs}{tb}  \\
  \pv{\tc:cbs}{tb}  & {}{\iso}{} & {\tt let}~\pv{cbs'}{tb'} = f~\pv{cbs}{tb}~{\tt in} 
                            ~\pv{\tc:cbs'}{tb'}
  \end{array}\!\!\right)
\end{array}
\]

The language is also expressible enough to write conventional
recursive (and higher-order) programs. We illustrate this
expressiveness using the usual $\mathtt{map}$ operation and an
accumulating variant $\mathtt{mapAccu}$:

\[
\begin{array}{l}
  {\tt map} : (a \iso b) \to ([a] \iso [b])
  \\
  \lambda g.\mu f.\left(
  \begin{array}{rcl}
    [] & {}\iso{} & []
    \\
    h:t & {}\iso{} & \letv{x}{g~h}{}
    \\
       &&\letv{y}{f~t}{x:y}
  \end{array}
         \right),
\end{array}
\]

\[
\begin{array}{l}
  {\tt mapAccu} :(a\otimes b \iso a\otimes c) \to
  (a\otimes [b] \iso a\otimes [c])
  \\
  \lambda g.\mu f.\left(
  \begin{array}{rcl}
    \pv{x}{[]} & {}\iso{} & \pv{x}{[]}
    \\
    \pv{x}{(h:t)} & {}\iso{} & {\tt let}~\pv{y}{h'} = g~\pv{x}{h}~{\tt in} \\
               && {\tt let}~\pv{z}{t'} = f~\pv{y}{t}~{\tt in}\\
               && \pv{z}{(h':t')}
  \end{array}
                  \right).
\end{array}
\]

\noindent The three examples {\tt cnot*}, {\tt map} and {\tt mapAccu}
use fixpoints which are clearly terminating in any finite context.
Indeed, the functions are structurally recursive. A formal definition
of this notion for the reversible language is as follows.

\begin{definition}\label{def:struct-rec}\rm
  Define a {\em structurally recursive type} as a type of the form
  $[a]\tensor b_1\tensor\ldots\tensor b_n$.  Let
  $\isoterm = \{ v_i \iso e_i ~|~ i\in I \}$ be an iso such that
  $
    f : a\iso b\entailiso\isoterm : a\iso c
  $
  where $a$ is a structurally recursive type.
  We say that $\mu f.\isoterm$ is {\em structurally recursive}
  provided that for each $i\in I$, the value $v_i$ is either of the
  form $\langle [], p_1, \ldots p_n\rangle$ or of the form
  $\langle h:t, p_1, \ldots p_n\rangle$.  In the former case, $e_i$
  does not contain $f$ as a free variable. In the latter case, $e_i$
  is of the form $C[f\pv{t}{p'_1,\ldots,p'_n}]$ where $C$ is a context of the form
  $C[-] ::= [-]\alt\letv{p}{C[-]}{t}\alt\letv{p}{t}{C[-]}$.
\end{definition}

\noindent This definition will be critical for quantum loops in the
next section.

\section{From Reversible Isos to  Quantum Control}

\begin{figure}[tb]
  \centering
  \begin{minipage}{.31\textwidth}
    \[\begin{blockarray}{cccc}
        & v_1 & v_2 & v_3 \\
        \begin{block}{c(ccc)}
          v'_1 ~& 1 & 0 & 0 \\
          v'_2 ~& 0 & 1 & 0 \\
          v'_3 ~& 0 & 0 & 1 \\
        \end{block}
      \end{blockarray}
    \]
    \vspace{-6ex}\caption{Classical iso}\label{fig:iso-id}\end{minipage}
  \qquad
  \begin{minipage}{.4\textwidth}
    \[
      \begin{blockarray}{cccc}
        & v_1 & v_2 & v_3 \\
        \begin{block}{c(ccc)}
          v'_1 ~& a_{11} & a_{12} & a_{13} \\
          v'_2 ~& a_{21} & a_{22} & a_{23} \\
          v'_3 ~& a_{31} & a_{32} & a_{33} \\
        \end{block}
      \end{blockarray}
    \]
    \vspace{-6ex}\caption{Quantum iso}\label{fig:iso-general}
  \end{minipage}
  \\
  \begin{minipage}{.4\textwidth}
    \[\begin{blockarray}{ccc}
        & \pv{\tc}{x} ~& \pv{\fc}{x} \\
        \begin{block}{c(cc)}
          \pv{\tc}{x} ~& \frac1{\sqrt2}{\tt Had} & \frac1{\sqrt2}{\tt Id} \\
          \pv{\fc}{x} ~& \frac1{\sqrt2}{\tt Had} & \frac{-1}{\sqrt2}{\tt
            Id} \\
        \end{block}
      \end{blockarray}
    \]
    \vspace{-6ex}\caption{Semantics of {\tt Gate}}\label{fig:sem-gate}
    \end{minipage}
\end{figure}

\noindent In the language presented so far, an iso $\isoterm:a\iso b$
describes a bijection between the set $\base{a}$ of closed values of
type $a$ and the set $\base{b}$ of closed values of type $b$. If one
regards $\base{a}$ and $\base{b}$ as the basis elements of some vector
space $\denot{a}$ and $\denot{b}$, the iso $\isoterm$ becomes a 0/1
matrix.

As an example, consider an iso $\isoterm$ defined using three clauses
of the form
\[
  \clauses{\clause{v_1}{v'_1}\clause{v_2}{v'_2}\clause{v_3}{v'_3}}.
  \]
From the exhaustivity and non-overlapping conditions derives the fact
that the space $\denot{a}$ can be split into the direct sum of the
three subspaces $\denot{a}_{v_i}$ ($i=1,2,3$) generated
by~$v_i$. Similarly, $\denot{b}$ is split into the direct sum of the
subspaces $\denot{b}_{v'_i}$ generated by~$v'_i$. One can therefore
represent $\isoterm$ as the matrix $\denot{\isoterm}$ in
Fig.~\ref{fig:iso-id}: The ``$1$'' in each column $v_i$ indicates to
which subspace $\denot{b}_{v'_j}$ an element of $\denot{a}_{v_i}$ is
sent to.

In Sec.~\ref{sec:1st-order-finite} we discussed the fact that
$v_i\bot v_j$ when $i\neq j$. This notation hints at the fact that
$\denot{a}$ and $\denot{b}$ could be seen as Hilbert spaces and the
mapping $\denot{\isoterm}$ as a unitary map from $\denot{a}$ to
$\denot{b}$. The purpose of this section is to extend and formalize
precisely the correspondence between isos and unitary maps.

\begin{figure}[t]
  \[
    \left\{
    \begin{array}{l}
      \clause{v_1~}{~a_{11}v'_1+a_{21}v'_2+a_{31}v'_3}\\
      \clause{v_2~}{~a_{12}v'_1+a_{22}v'_2+a_{23}v'_3}\\
      \clause{v_3~}{~a_{31}v'_1+a_{32}v'_2+a_{33}v'_3}
    \end{array}
  \right\}
\]
\caption{Illustrating the generalization of clauses}\label{fig:iso-gen-ex}
\end{figure}
The definition of clauses is extended following this idea of seeing
isos as unitaries, and not only bijections on basis elements of the
input space. We
therefore essentially propose to generalize the clauses to complex,
linear combinations of values on the right-hand-side, such as shown in
Figure~\ref{fig:iso-gen-ex}, with the side conditions on that the matrix of
Fig.~\ref{fig:iso-general} is unitary. We define in
Sec.~\ref{sec:unitlang} how this extends to second-order.

\subsection{Extending the Language to Linear Combinations of Terms}
\label{sec:unitlang}

The quantum unitary language extends the reversible language from the
previous section by closing extended values and terms under complex,
finite linear combinations. For example, if $v_1$ and $v_2$ are values
and $\alpha$ and $\beta$ are complex numbers,
$\alpha\cdot v_1 + \beta\cdot v_2$ is now an extended value.

Several approaches exist for performing such an extension. One can
update the reduction strategy to be able to reduce these sums and
scalar multiplications to normal
forms~\cite{lineal,alglambdacalcreview}, or one can instead consider
terms modulo the usual algebraic
equalities~\cite{Vaux09,alglambdacalcreview}: this is the strategy we
follow for this paper.

When extending a language to linear combination of terms in a naive
way, this added structure might generate inconsistencies in the
presence of unconstrained
fixpoints~\cite{Vaux09,lineal,alglambdacalcreview}. The weak condition
on termination we imposed on fixpoints in the classical language was
enough to guarantee reversibility. With the presence of linear
combinations, we want the much stronger guarantee of unitarity. For
this reason, we instead impose fixpoints to be {\em structurally
  recursive}.

The quantum unitary language is defined by allowing sums of terms and
values and multiplications by complex numbers: if $t$ and $t'$ are
terms, so is $\alpha\cdot t + t'$. Terms and values are taken modulo
the equational theory of modules. We furthermore consider the value
and term constructs $\pv{-}{-}$, $\letv{p}{-}{-}$, $\inl(-)$,
$\inr(-)$ distributive over sum and scalar multiplication.
We do {\em not} however take iso-constructions as distributive over
sum and scalar multiplication:
$\clauses{\clause{v_1}{\alpha v_2 + \beta v_3}}$ is {\em not} the same
thing as
$\alpha\clauses{\clause{v_1}{v_2}} +
\beta\clauses{\clause{v_1}{v_3}}$.
This is in the spirit of Lineal~\cite{lineal,linvec}.

Formally, the quantum unitary language is defined as follows.
\begin{alignat*}{100}
&\text{Val }\& ~\text{term types} \quad& a, b &&&::=~ &&  
                        \mathbb{1} \alt a \oplus b \alt a \otimes b \alt [a]
                          \\
&\text{Iso types} & T &&&::=&&  a \iso b \alt (a \iso b) \to T \\[1.5ex]
&\text{Pure values} & v &&&::=&&
                       () \alt x \alt \inl{v} \alt \inr{v} \alt
                        \pv{v_1}{v_2}\\
&\text{Combination of values} & e &&&::=&&
                       v \alt e_1 + e_2 \alt \alpha e\\
&\text{Products} & p &&&::=&&
                       () \alt x \alt \pv{p_1}{p_2}\\
&\text{Extended Values} & e &&&::=&& v \alt
                                 \letv{p_1}{\isoterm~p_2}{e}\\
&\text{Isos} & \isoterm &&&::=&& 
                          \clauses{\clause{v_1}{e_1}\clause{v_2}{e_2}~\ldots}
                          \alt \lambda f.\isoterm \alt \\
&            &          &&&   && \mu f.\isoterm \alt f \alt 
                                          \isoterm_1\,\isoterm_2 \\[2ex]
&\text{Terms}  & t &&&::=&& () \alt x \alt \inl{t} \alt \inr{t} \alt
                        \pv{t_1}{t_2} \alt \\ 
&              &   &&&   && \isoterm~t \alt \letv{p}{t_1}{t_2}\alt t_1 + t_2 \alt \alpha\cdot t.
\end{alignat*}
The scalar $\alpha$ ranges over complex numbers. Extended terms and
values are considered modulo associativity and commutativity of the
addition, and modulo the equational theory of modules:
\begin{align*}
  \alpha\cdot(e_1 + e_2) &= \alpha\cdot{}e_1 + \alpha\cdot{}e_2 &
  1\cdot{}e &= e\\
  \alpha\cdot{}e + \beta\cdot{}e &= (\alpha+\beta)\cdot{}e&
  \alpha\cdot(\beta\cdot{}e) &= (\alpha\beta)\cdot e\\
  0\cdot e_1 + e_2 &= e_2
\end{align*}
We furthermore consider the value and term constructs $\pv{-}{-}$,
$\letv{p}{-}{-}$, $\inl(-)$, $\inr(-)$ distributive over sum and
scalar multiplication.

The typing rules for terms and extended values are updated
as shown in
Table~\ref{tab:quantum-types}.
We only allow linear combinations of
terms and values of the same type and of the same free variables. An
iso is now not only performing an ``identity'' as in
Figure~\ref{fig:iso-id} but a true unitary operation.
Finally, fixpoints are now required to be {\em structurally
  recursive}, as introduced in Definition~\ref{def:struct-rec}.

\begin{table}[tb]
\[
\begin{array}{c}
\infer{\Delta;\Psi\entailval\alpha\cdot t:a }{\Delta;\Psi\entailval t:a}
\quad
\infer{
  \Delta;\Psi\entailval t_1+t_2 : a 
}{
  \Delta;\Psi\entailval t_1 : a
  & 
  \Delta;\Psi\entailval t_2 : a
}
\\[2ex]
\infer{
  \Psi\entailiso 
  \left\{
  \begin{array}{ccc}
    v_1 & \iso & a_{11}\cdot e_1 + \cdots + a_{1n}\cdot e_n
    \\
        &\ldots&
    \\
    v_n & \iso & a_{n1}\cdot e_1 + \cdots + a_{nn}\cdot e_n
  \end{array}\right\}
  : a \iso b.
}{
  \begin{array}{ll@{~~~}l}
    \Delta_1;\Psi \entailval v_1 : a 
    &
      \ldots
    &
      \Delta_n;\Psi\entailval v_n : a
      \\
    \Delta_1;\Psi \entailval e_1 : b 
    &
      \ldots
    &
      \Delta_n;\Psi\entailval e_n : b
    \\
    \OD{a}\{v_1,\ldots,v_n\}
    &&
       \ODe{b}\{e_1,\ldots,e_n\}
  \end{array}
  &
  \begin{pmatrix}
    a_{11} & \cdots &a_{1n}\\
    \vdots & & \vdots\\
    a_{n1} & \cdots & a_{nn}
  \end{pmatrix}  
  \text{ is unitary}
}
\\[6ex]
\inferrule{
  \Psi, f:a\iso b\entailiso \isoterm : (a_1\iso b_1)\to\cdots \to(a_n\iso
  b_n)\to(a\iso b)
  \\
  \text{$\mu f.\isoterm$ is structurally recursive}
}
{
  \Psi\entailiso \mu f.\isoterm : (a_1\iso b_1)\to\cdots \to(a_n\iso b_n)\to(a\iso b)
}
\end{array}
\]
\caption{Typing rules for the quantum extension}
\label{tab:quantum-types}
\end{table}

The reduction is updated to stay deterministic in this extended
setting. It is split into two parts: the reduction of pure terms,
i.e. non-extended terms or values, and linear combinations thereof.
\begin{itemize}
\item Pure terms and values reduces using the reduction rules found in
  Table~\ref{tab:reduction}. We do not extend applicative contexts to
  linear combinations. The only slightly modified rule is the rule
  (IsoApp): the $t_i$ might now be a linear combination of pure terms
  $\sum_j\alpha_j\cdot t'_j$. We define $\sigma(t_i)$ and
  $\sum_j\alpha_j\cdot\sigma(t'_j)$. Because of the constraint on the
  typing rule for linear combinations, as long as $t_i$ is well-typed
  then the substitution is well-defined on all the $t'_j$.
\item Consider the linear combination of pure terms
  $\sum_i\alpha_i\cdot t_i + \sum_j\beta_j\cdot t'_j$, where the $t_j$
  are in normal form but the $t_i$ are not: $t_i\to t''_i$. Then 
  \[
    \sum_i\alpha_i\cdot t_i + \sum_j\beta_j\cdot t'_j
    \to
    \sum_i\alpha_i\cdot t''_i + \sum_j\beta_j\cdot t'_j
  \]
  Note that this extended reduction relation is deterministic.
\end{itemize}

\begin{example}
\label{ex:had}\label{sec:example-quantum}
This allows one to define an iso behaving as the Hadamard gate, or a
slightly more complex iso conditionally applying another iso, whose
behavior as a matrix is shown in Fig.~\ref{fig:sem-gate}.

\[
\begin{array}{l}
       \mathtt{Had} : \boolt \iso \boolt \\
       \left(
       \begin{array}{r@{~}c@{~}l}
         \tc & {}\iso{} & \frac1{\sqrt2}\tc + \frac1{\sqrt2}\fc
         \\
         \fc & {}\iso{} & \frac1{\sqrt2}\tc - \frac1{\sqrt2}\fc
       \end{array}
      \right)\!,
    \end{array}
  \]
  \[
    \begin{array}{l}
       \mathtt{Gate} : \boolt\otimes\boolt \iso \boolt\otimes\boolt \\
       \left(
       \begin{array}{r@{~}c@{~}l}
         \pv{\tc}{x} & {}\iso{} & \letv{y}{{\tt Had}\,x}{\frac1{\sqrt2}\pv{\tc}{y} + \frac1{\sqrt2}\pv{\fc}{y}}
         \\
         \pv{\fc}{x} & {}\iso{} & \letv{y}{{\tt Id}\,x}{~~\frac1{\sqrt2}\pv{\tc}{y} - \frac1{\sqrt2}\pv{\fc}{y}}
       \end{array}\right)\!.
     \end{array}
  \]
\end{example}

With this extension to linear combinations of terms, one can
characterize normal forms as follows.

\begin{lemma}[Structure of the normal forms]
  \label{lem:isonorm1}
  Let $\isoterm$ be such that
  $\entailiso\isoterm : a\iso b$. For all closed values $v$ of type
  $a$, the term $\isoterm\,v$ rewrites to a normal form
  $
    \sum_{i=1}^N\alpha_i\cdot w_i
  $
  where $N<\infty$, each $w_i$ is a closed value of type $b$ and
  $\sum_i|\alpha_i| = 1$.
\end{lemma}

\begin{proof}
  The fact that $\isoterm\,v$ converges to a normal form is a
  corollary of the fact that we impose structural recursion on
  fixpoints. The property of the structure of the normal form is then
  proven by induction on the maximal number of steps it takes to reach
  it. It uses the restriction on the introduction of sums in the
  typing rule for clauses in isos and the determinism of the
  reduction.
\end{proof}

In the classical setting, isos describe bijections between sets of
closed values: it was proven by considering the behavior of an iso
against its inverse. In the presence of linear combinations of terms,
we claim that isos describe more than bijections: they describe
unitary maps. In the next section, we discuss how types can be
understood as Hilbert spaces (Sec.~\ref{sec:types-hilb}) and isos as
unitary maps (Secs~\ref{sec:isos-blmpas} and~\ref{sec:isos-umaps}).

\subsection{Modeling Types as Hilbert Spaces}
\label{sec:types-hilb}

By allowing complex linear combinations of terms, closed normal forms
of finite types such as $\boolt$ or $\boolt\otimes\boolt$ can be
regarded as complex vector spaces with basis consisting of closed
values. For example, $\boolt$ is associated with
$\denot{\boolt}=\{\alpha\cdot\tc +
\beta\cdot\fc~|~\alpha,\beta\in\Cx\}\equiv\Cx^2$.  We can consider
this space as a complex Hilbert space where the scalar product is
defined on basis elements in the obvious way: $\scalprod{v}{v} = 1$
and $\scalprod{v}{w} = 0$ if $v\neq w$. The map ${\tt Had}$ of
Ex.~\ref{ex:had} is then effectively a unitary map on the space
$\denot{\boolt}$.

The problem comes from lists: the type $[\mathbb{1}]$ is inhabited by
an infinite number of closed values: $[]$, $[()]$, $[(),()]$,
$[(),(),()]$,\ldots To account for this case, we need to consider
infinitely dimensional complex Hilbert spaces. In general, a complex
Hilbert space~\cite{HSBook} is a complex vector space endowed with a
scalar product that is complete with respect the distance induced by
the scalar product. The completeness requirement implies for example
that the infinite linear combination
$[] + \frac12\cdot[()] + \frac14[(),()] + \frac18[(),(),()] + \cdots$
needs to be an element of $\denot{[\boolt]}$. To account for these
limit elements, we propose to use the standard~\cite{HSBook} Hilbert
space $\ell^2$ of infinite sequences.

\begin{definition}\label{def:hilb}\rm
  Let $a$ be a value type. As before, we write $\base{a}$ for the set
  of closed values of type $a$, that is,
  $\base{a} = \{ v ~|~ \entailval v:a \}$. The {\em span of $a$} is
  defined as the Hilbert space $\denot{a} = \ell^2(\base{a})$
  consisting of sequences $(\phi_v)_{v\in\base{a}}$ of complex numbers
  indexed by $\base{a}$ such that
  $\sum_{v\in\base{a}}|\phi_v|^2<\infty$. The scalar product on this
  space is defined as
  $\scalprod{(\phi_v)_{v\in\base{a}}}{(\psi_v)_{v\in\base{a}}} =
  \sum_{v\in\base{a}} \overline{\phi_v}\psi_v$.
\end{definition}

We shall use the following conventions. A closed value $v$ of
$\denot{a}$ is identified with the sequence
$(\delta_{v,v'})_{v'\in\base{a}}$ where $\delta_{v,v} = 1$ and
$\delta_{v,v'}=0$ if $v\neq v'$. An element $(\phi_v)_{v\in\base{a}}$
of $\denot{a}$ is also written as the infinite, formal sum
$\sum_{v\in\base{a}}\phi_v\cdot v$.

\subsection{Modeling Isos as Bounded Linear Maps}
\label{sec:isos-blmpas}

We can now define what is the linear map associated to an iso.

\begin{definition}\label{def:linmap}\rm
  For each closed iso $\entailiso\isoterm : a\iso b$ we define
  $\denot{\isoterm}$ as the linear map from $\denot{a}$ to $\denot{b}$
  sending the closed value $v:a$ to the normal form of 
  $\isoterm\,v:b$ under the rewrite system.
\end{definition}

In general, the fact that $\denot{\isoterm}$ is well-defined is not
trivial. If it is formally stated in Theorem~\ref{th:unitwdef}, we can
first try to understand what could go wrong.
The problem comes from the fact that the space $\denot{a}$ is not
finite in general. Consider the iso
${\tt map}~{\tt Had} : [\boolt]\iso[\boolt]$. Any closed value
$v:[\boolt]$ is a list and the term $({\tt map}~{\tt Had})\,v$
rewrites to a normal form consisting of a linear combination of lists.
Denote the linear combination associated to $v$ with $L_v$.
An element of $\denot{[\boolt]}$ is a sequence
$ \phi = (\phi_v)_{v\in\base{[\boolt]}} $. From
Definition~\ref{def:linmap}, the map $\denot{\isoterm}$ sends the
element $\phi\in\denot{[\boolt]}$ to
$ \sum_{v\in\base{[\boolt]}} \phi_v \cdot L_v. $ This is an infinite
sum of sums of complex numbers: we need to make sure that it is
well-defined: this is the purpose of the next result. Because of the
constraints on the language, we can even show that it is a {\em
  bounded} linear map.

In the case of the map ${\tt map}~{\tt Had}$, we can understand why it
works as follows. The space $\denot{[\boolt]}$ can be decomposed as
the direct sum $\sum_{i=0}^\infty E_i$, where $E_i$ is generated with
all the lists in $\boolt$ of size $i$. The map ${\tt map}~{\tt Had}$
is acting locally on each finitely-dimensional subspace $E_i$. It is
therefore well-defined. Because of the unitarity constraint on the
linear combinations appearing in {\tt Had}, the operation performed by
${\tt map}~{\tt Had}$ sends elements of norm 1 to elements of norm 1.
This idea can be formalized and yield the following theorem.

\begin{theorem}
  \label{th:unitwdef}
  For each closed iso $\entailiso\isoterm : a\iso b$ the linear map
  $\denot{\isoterm} : \denot{a}\to\denot{b}$ is well-defined and
  bounded.
\end{theorem}

\begin{proof}
  Consider a general element $e = (e_v)_{v\in\base{a}}$ of
  $\denot{a}$. Using Lemma~\ref{lem:isonorm1}, to each $v\in\base{a}$
  one can attach a finite linear combination
  \[
    W_v = \sum_{i=1}^{N_v}\alpha^v_i\cdot w^v_i
  \]
  such that $\isoterm\,v$ rewrites to $W_v$, with
  $\sum_i|\alpha^v_i|=1$.
  By definition, $\denot{\isoterm}(e)$ is a sequence indexed by
  $\base{b}$ where for all $w$ in $\base{b}$,
  \[
    (\denot{\isoterm}(e))_w = 
    \left(\sum_{v\in\base{a}} e_v\cdot W_v\right)_w
    =\sum_{v\in\base{a}} e_v\cdot \left(W_v\right)_w
  \]
  This series is absolutely converging, and therefore
  well-defined. Indeed,
  \begin{alignat*}{100}
    \sum_{v\in\base{a}} \left|e_v\cdot \left(W_v\right)_w\right|
    &
    = \sum_{v\in\base{a}} |e_v|\cdot \left|\left(W_v\right)_w\right|
    \\
    &\leq \sum_{v\in\base{a}} |e_v| & \text{because $|(W_v)_w|\leq 1$}
    \\
    &\leq \infty & \text{because $e\in\base{a}$.}
  \end{alignat*}
  To show that $\denot{\isoterm}(e)$ is indeed an element of
  $\base{b}$, we need to show that its norm is well-defined.  We show
  this by bounding it with the norm of the vector $e$.
  \begin{alignat*}{100}
    &\sum_{w\in\base{b}}|(\denot{\isoterm}(e))_w|
    \\
    & = \sum_{w\in\base{b}}\left|\sum_{v\in\base{a}} e_v\cdot
      \left(W_v\right)_w\right|
    \\
    & \leq \sum_{w\in\base{b}}\sum_{v\in\base{a}} \left|e_v\cdot
      \left(W_v\right)_w\right|
    \\
    & = \sum_{w\in\base{b}}\sum_{v\in\base{a}} |e_v|\cdot
      |\left(W_v\right)_w|
    \\
    & = \sum_{v\in\base{a}}\sum_{w\in\base{b}} |e_v|\cdot
    |\left(W_v\right)_w|
    \\
    & = \sum_{v\in\base{a}}|e_v|\cdot\left(\sum_{w\in\base{b}} 
    |\left(W_v\right)_w|\right)
    \\
    & = \sum_{v\in\base{a}}|e_v|\cdot\left(\sum_{i=1}^{N_v} 
    |\alpha^v_i|\right)
    \\
    & = \sum_{v\in\base{a}}|e_v|
    \\
    & = \|e\|.
  \end{alignat*}
  This concludes the proof that the linear map $\denot{\isoterm}$ is
  not only well-defined between the vector spaces $\denot{a}$ and
  $\denot{b}$ but also bounded.
\end{proof}

\subsection{Modeling Isos as Unitary Maps}
\label{sec:isos-umaps}

In this section, we show that not only closed isos can be modeled as
bounded linear maps, but that these linear maps are in fact unitary
maps. The problem comes from fixpoints. We first consider the case of
isos written without fixpoints, and then the case with fixpoints.

\paragraph{Without recursion.}
The case without recursion is relatively easy to treat, as the linear
map modeling the iso can be compositionally constructed out of
elementary unitary maps.

\begin{theorem}\label{th:unitnorec}
  Given a closed iso $\entailiso\isoterm : a\iso b$ defined without
  the use of recursion, the linear map
  $\denot{\pi} : \denot{a} \to \denot{b}$ is unitary.
\end{theorem}

\begin{proof}[Proof sketch]
The proof of the theorem relies on the fact that to each closed iso
$\entailiso\isoterm : a\iso b$ one can associate an operationally
equivalent iso $\entailiso\isoterm' : a\iso b$ that does not use
iso-variables nor lambda-abstractions.
Such an iso $\isoterm$ is necessarily of the canonical form
\begin{equation}\label{eq:canonical-form}
\left(
  \begin{array}{rcl}
    v_1 &\iso& {\tt let}\,{p_{11}}=\isoterm_{11}\,p'_{11}~{\tt in}\\
        &&  {\tt let}\,{p_{12}}=\isoterm_{12}\,p'_{12}~{\tt in}\\
        &&  \cdots\\
        && {\tt in}~\alpha_{11}w_1 + \cdots + \alpha_{1n}w_n
    \\
    \vdots &\vdots& \vdots
    \\
    v_n &\iso& {\tt let}\,{p_{n1}}=\isoterm_{n1}\,p'_{n1}~{\tt in}\\
        &&  {\tt let}\,{p_{n2}}=\isoterm_{n2}\,p'_{n2}~{\tt in}\\
        &&  \cdots\\
        && {\tt in}~\alpha_{n1}w_1 + \cdots + \alpha_{nn}w_n
  \end{array}
\right)
\end{equation}
where all the $\isoterm_{ij}$ are also of this canonical form. We
define the {\em applicative depth} of such an iso as follows: it is of
depth~0 if its definition does not mention any $\isoterm_{ij}$, and if
it does, its depth is $1$ plus the maximum depth of its
$\isoterm_{ij}$.

  We then prove the theorem by induction on the depth of the iso $\pi$. If it
  is of depth 0, then it is done because of the fact that
  $(\alpha_{ij})_{i,j}$ forms a unitary matrix.

  If it is of depth $n+1$, and if the result is true for all isos of
  depth less or equal to $n$, then consider the factorization of the
  iso of Equation~\eqref{eq:canonical-form} into $\isoterm_{\it
    straight}$ followed by $\isoterm_{\it rotate}$, where
  \[
      \begin{array}{@{}l@{}}
    \isoterm_{\it straight} =\\[1.2ex] 
\left(
  \begin{array}{r@{~\iso~}l}
    v_1 & {\tt let}\,{p_{11}}=\isoterm_{11}\,p'_{11}~{\tt in}~
          {\tt let}\,{p_{12}}=\isoterm_{12}\,p'_{12}~{\tt in}~
          \cdots ~{\tt in}~w_1
    \\
    \cdots  & \cdots
    \\
    v_n & {\tt let}\,{p_{n1}}=\isoterm_{n1}\,p'_{n1}~{\tt in}~
          {\tt let}\,{p_{n2}}=\isoterm_{n2}\,p'_{n2}~{\tt in}~
          \cdots ~{\tt in}~w_n
  \end{array}
\right),
      \end{array}
    \]
    \[
    \begin{array}{@{}l@{}}
    \isoterm_{\it rotate} = \\[1.2ex]
\left(
  \begin{array}{r@{~\iso~}l}
    w_1 & \alpha_{11}w_1 + \cdots + \alpha_{1n}w_n
    \\
    \cdots & \cdots
    \\
    w_n & \alpha_{n1}w_1 + \cdots + \alpha_{nn}w_n
  \end{array}
\right).
      \end{array}
    \]
The iso $\denot{\pi_{\it rotate}}$ is a unitary because the
  $(\alpha_{ij})_{i,j}$ forms a unitary matrix. For the iso
  $\denot{\pi_{\it straight}}$, since the $v_i$'s form an exhaustive
  and non-overlapping coverage of $a$, the space $\denot{a}$ can be
  orthogonally decomposed as
  $\denot{v_i}\oplus\cdots\oplus\denot{v_n}$ where $\denot{v_i}$ is
  the subspace corresponding to the closed values matching
  $v_i$. Similarly, $\denot{b}$ can be orthogonally decomposed into
  $\denot{w_1}\oplus\cdots\oplus\denot{w_n}$. The map
  $\denot{\pi_{\it straight}}$ is then sending each subspace
  $\denot{v_i}$ to the subspace $\denot{w_i}$. The corresponding
  operation is the composition of $\denot{\pi_{i1}}$ with
  $\denot{\pi_{i2}}$ with $\denot{\pi_{i3}}$ with \ldots, each being
  unitaries by induction hypothesis.
  Summing up, $\denot{\pi_{\it straight}}$ is unitary because it is
  the (orthogonal) sum of compositions of unitaries: $\denot{\pi}$ is
  therefore unitary.
\end{proof}

As an illustration, the semantics of {\tt Gate} of
Example~\ref{sec:example-quantum} is given in
Figure~\ref{fig:sem-gate}.

\paragraph{Isos with structural recursion.}
When considering fixpoints, we cannot rely anymore on this finite
compositional construction: the space $\denot{a}$ cannot anymore be
regarded as a {\em finite} sum of subspaces described by each clause.

We therefore need to rely on the formal definition of unitary maps in
general, infinite Hilbert spaces. On top of being bounded linear, a
map $\denot{\isoterm}:\denot{a}\to\denot{b}$ is unitary if (1) it
preserves the scalar product:
$ \scalprod{\denot{\isoterm}(e)}{\denot{\isoterm}(f)} =
\scalprod{e}{f} $ for all $e$ and $f$ in $\denot{a}$ and (2) it is
surjective.

\begin{theorem}\label{th:unitwithrec}
  Given a closed iso $\entailiso\isoterm : a\iso b$ that can use
  structural recursion, the linear map
  $\denot{\pi} : \denot{a} \to \denot{b}$ is unitary.\qed
\end{theorem}

The proof uses the idea highlighted in Sec.~\ref{sec:isos-umaps}: for
a structurally recursive iso of type $[a]\tensor b \iso c$, the
Hilbert space $\denot{[a]\tensor b}$ can be split into a canonical
decomposition $E_0\oplus E_1\oplus E_2\oplus\cdots$, where $E_i$
contains only the values of the form $\pv{[x_1\ldots x_i]}{y}$,
containing the lists of size $i$. On each $E_i$, the iso is equivalent
to an iso without structural recursion.

\begin{proof}
  First note that for any given {\em finite} set $(v_k)_k$ of values of
  type $a$, there exists an unfolding $\isoterm'$ of $\isoterm$ such
  that for all $k$, $ \isoterm' v_k \to\ldots \to W_k $ with $W_k$ a
  normal form of type $b$, and such that there are no rewrite step
  performing a fixpoint unfolding.
  (Indeed, for each $v_k$ there is such an unfolding: there exists a
  maximal unfolding that matches all of them).

  We then ensure that $\denot{\isoterm'}$ acts in a unitary way on a
  subspace of $\denot{a}$, i.e. that the trick of the decomposition we
  used for the finite case in Theorem~\ref{th:unitnorec} can be
  applied here.
  This is the purpose of the condition of structural recursion.
  
  The point is that now for all $n$, the set $\base{a}$ can be
  decomposed as the disjoint union of the values with lists of length
  smaller than $n$ and the ones of length larger than $n$:
  $
    \base{a}^<_n =
     \{ v \tensor [e_1\ldots e_k]  |  v : a_1  \text{~and~}
     k \leq n \text{~and~}  \forall i, e_i : a_2 \}
   $
   and
   $
     \base{a}^>_n =
     \{ v \tensor [e_1\ldots e_k]  |  v : a_1  \text{~and~}
     k > n \text{~and~}  \forall i, e_i : a_2 \}.
   $
   For a fixed $n$, the fixpoint $\mu f.\isoterm$ can be unfolded so
   that if $v \in \base{a}^<_n$ the rewriting of $(\mu f.\isoterm) v$
   to a normal form does not need to contain any unfolding of $f$.

   Therefore, the Hilbert space $\denot{a}$ is the disjoint sum
   $\ell^2(\base{a}^<_n) \oplus \ell^2(\base{a}^>_n)$
   and the action of $\denot{\mu f.\isoterm}$ on this space can be
   decomposed as
   $
     \denot{\mu f.\isoterm}^<_n + \denot{\mu f.\isoterm}^>_n
   $
   where $\denot{\mu f.\isoterm}^<_n$ has a representation as a
   composition of unitaries, invoking Theorem~\ref{th:unitnorec}.

   The corollary is that if $e$ and $f$ are elements of $\denot{a}$
   with support in $\base{a}^<_n$ we effectively have
   $
     \scalprod{\denot{\mu f.\isoterm}(e)}{\denot{\mu f.\isoterm}(f)}
     =
     \scalprod{e}{f}.
   $
   Now, consider $e$ and $f$ two {\em general} elements of $\denot{a}$.

   We claim that we still have the above equality, and we obtain it by
   a limit argument.
   Indeed, construct $e_n$ and $f_n$ the elements built from $e$ and
   $f$ by restricting their support to $\base{a}^<_n$.
   Because of the norm condition on $e$ and $f$, the errors between
   $\scalprod{e}{f}$ and $\scalprod{e_n}{f_n}$ and the error between
   $\scalprod{\denot{\mu f.\isoterm}(e)}{\denot{\mu f.\isoterm}(f)}$
   and
   $\scalprod{\denot{\mu f.\isoterm}(e_n)}{\denot{\mu f.\isoterm}(f_n)}$
   goes to zero as $n$ goes to infinity.

   The last thing to check is that this operator is indeed
   surjective. But then the same argument as for classical fixpoints
   can be used: as long as the syntactic inverse is proved to be
   total, we retrieve surjectivity.
\end{proof}

\section{Conclusion}

In this paper, we proposed a reversible language amenable to quantum
superpositions of values. The language features a weak form of
higher-order that is nonetheless expressible enough to get interesting
maps such as generalized Toffoli operators. We sketched how this language effectively encodes bijections in
the classical case and unitary operations in the quantum case. It
would be interesting to see how this relates to join inverse categories~\cite{catflow,KAARSGAARD201733}.

In the vectorial extension of the language we have the same control as
in the classical, reversible language. Tests are captured by clauses,
and naturally yield quantum tests: this is similar to what can be
found in QML~\cite{qml,qalternation}, yet more general since the QML
approach is restricted to {\tt if-then-else} constructs. The novel
aspect of quantum control that we are able to capture here is a notion
of {\em quantum loops}. These loops were believed to be hard, if not
impossible. What makes it work in our approach is the fact that we are
firmly within a closed quantum system, without measurements. This
makes it possible to only consider unitary maps and frees us from the
L\"ower order on positive matrices~\cite{qalternation}. As we restrict
fixpoints to structural recursion, valid isos are regular enough to
capture unitarity. Ying~\cite{yingbook} also proposes a framework for
quantum while-loops that is similar in spirit to our approach at the
level of denotations: in his approach the control part of the loops is
modeled using an external systems of ``coins'' which, in our case,
correspond to conventional lists. Reducing the manipulation of this
external coin system to iteration on lists allowed us to give a simple
operational semantics for the language.


\begin{thebibliography}{10}

\bibitem{wootters82single}
Wootters, W.K., Zurek, W.H.:
\newblock A single quantum cannot be cloned.
\newblock Nature \textbf{299} (October 1982)  802--803

\bibitem{NielsenChuang}
Nielsen, M.A., Chuang, I.L.:
\newblock Quantum Computation and Quantum Information.
\newblock Cambridge University Press (2002)

\bibitem{quipper}
Green, A.S., Lumsdaine, P.L., Ross, N.J., Selinger, P., Valiron, B.:
\newblock Quipper: A scalable quantum programming language.
\newblock In: \textit{Proc. PLDI'13. (2013)}  333--342

\bibitem{qwire}
Paykin, J., Rand, R., Zdancewic, S.:
\newblock {QWIRE}: A core language for quantum circuits.
\newblock In: \textit{Proc. POPL'17. (2017)}  846--858

\bibitem{qml}
Altenkirch, T., Grattage, J.:
\newblock A functional quantum programming language.
\newblock In: \textit{Proc. LICS'05. (2005)}  249--258

\bibitem{qalternation}
Badescu, C., Panangaden, P.:
\newblock Quantum alternation: Prospects and problems.
\newblock In: \textit {Proc. QPL'15. (2015)}  33--42

\bibitem{yingbook}
Ying, M.:
\newblock Foundations of Quantum Programming.
\newblock Elsevier Science (2016)

\bibitem{selinger04quantum}
Selinger, P.:
\newblock Towards a quantum programming language.
\newblock Mathematical Structures in Computer Science \textbf{14}(4) (August
  2004)  527--586

\bibitem{qarrow}
Vizzotto, J.K., Altenkirch, T., Sabry, A.:
\newblock Structuring quantum effects: superoperators as arrows.
\newblock Mathematical Structures in Computer Science \textbf{16}(3) (2006)
  453--468

\bibitem{theseus}
James, R.P., Sabry, A.:
\newblock Theseus: A high-level language for reversible computation.
\newblock In: Reversible Computation, Booklet of work-in-progress and short
  reports. (2016)

\bibitem{linvec}
Arrighi, P., D{\'i}az-Caro, A., Valiron, B.:
\newblock The vectorial lambda-calculus.
\newblock Information and Computation \textbf{254}(1) (2017)  105--139

\bibitem{lineal}
Arrighi, P., Dowek, G.:
\newblock Lineal: A linear-algebraic lambda-calculus.
\newblock Logical Methods in Computer Science (2013)

\bibitem{Vaux09}
Vaux, L.:
\newblock The algebraic lambda calculus.
\newblock Mathematical Structures in Computer Science \textbf{19}(5) (2009)
  1029--1059

\bibitem{catflow}
Gl{\"u}ck, R., Kaarsgaard, R.:
\newblock A categorical foundation for structured reversible flowchart
  languages: Soundness and adequacy.
\newblock Available on arXiv:1710.03666 [cs.PL] (2017)

\bibitem{KAARSGAARD201733}
Kaarsgaard, R., Axelsen, H.B., Gl{\"u}ck, R.:
\newblock Join inverse categories and reversible recursion.
\newblock Journal of Logical and Algebraic Methods in Programming \textbf{87}
  (2017)  33--50

\bibitem{tonder04lambda}
van Tonder, A.:
\newblock A lambda calculus for quantum computation.
\newblock SIAM Journal of Computing \textbf{33}(5) (2004)  1109--1135

\bibitem{shortversion} Sabry, A., Valiron, B., Vizzotto, J.K.:
  \newblock From symmetric pattern-matching to quantum control
  \newblock To appear in: {\em Proc. FoSSaCS'18 (2018)}.

\bibitem{alglambdacalcreview}
Assaf, A., D{\'i}az-Caro, A., Perdrix, S., Tasson, C., Valiron, B.:
\newblock Call-by-value, call-by-name and the vectorial behaviour of the
  algebraic $\lambda$-calculus.
\newblock {Logical Methods in Computer Science} \textbf{10:4}(8) (December
  2014)

\bibitem{HSBook}
Young, N.:
\newblock An Introduction to Hilbert Space.
\newblock Cambridge University Press (1988)

\end{thebibliography}
\end{document}